\documentclass[12pt,a4paper]{article}
\usepackage{amsmath}
\usepackage{amsthm}
\usepackage{amssymb}
\usepackage{mathrsfs}
\usepackage{enumitem}
\usepackage[left=2.5cm,right=2.5cm,top=2.5cm,bottom=2.5cm]{geometry}

\usepackage{natbib}
\usepackage[breaklinks=true,colorlinks,dvipdfm,citecolor=blue]{hyperref}
\usepackage[all]{hypcap}
\usepackage{graphicx}

\usepackage{booktabs}
\usepackage{color}

\usepackage{titling}

\thanksmarkseries{arabic}
\date{\today}
\author{Katar\'{\i}na Burclov\'{a}\thanks{katarina.burclova@fmph.uniba.sk} and Andrej P\'{a}zman\thanks{pazman@fmph.uniba.sk}\\
Comenius University Bratislava\thanks{Faculty of Mathematics, Physics and Informatics, Mlynsk\'{a} dolina, 842 48 Bratislava, Slovak Republic}
}
\title{Optimal design of experiments via linear programming}

\newtheorem{theorem}{Theorem}

\theoremstyle{definition}
\newtheorem{remark}{Remark}

\theoremstyle{definition}
\newtheorem{example}{Example}

\begin{document}

\maketitle
\begin{abstract}
We investigate the possibility of extending some results
of \cite{PP14} to a larger set of optimality criteria. Namely, in a linear
regression model the problem of computing $D$-, $A$-, $E_k$-optimal designs, of
combining these optimality criteria, and the ``criterion robust'' problem of
\cite{H04} are reformulated here as ``infinite-dimensional'' linear programming problems.
Approximate optimum designs can then be computed by a modified cutting-plane
method, and this is checked on examples. Finally, the expressions for these
criteria are reformulated in terms of the response function of an even
nonlinear model.
\end{abstract}
\section{Introduction} 
\label{sec:intro}
We consider a regression model 
\[
y_i=\eta \left( x_i,\theta \right) +\varepsilon _i,\quad i=1,\hdots,N ,
\]
where $y_i$ are observed variables, $\varepsilon _i$ are observation errors,
which satisfy $\mathbb{E}\left( \varepsilon _i\right) =0$, and $Var\left(
\varepsilon _i\right) =\sigma ^2,\,Cov\left( \varepsilon _i,\varepsilon
_j\right) =0$ for $i\neq j$, $\sigma ^2$ is not supposed to be known. The
value of $\theta $ is a priori restricted to a parameter space $\Theta $.  In a
vector notation the model is 
\begin{eqnarray*}
{y} &=&{\eta }_X\left( \theta \right) +{\varepsilon }, \\
\mathbb{E}\left( {\varepsilon }\right) &=&{0},Var\left( {
\varepsilon }\right) =\sigma ^2{I}.
\end{eqnarray*}
Here $X=\left( x_1,\hdots,x_N\right) $ is the exact design with points $x_i\in 
\mathcal{X}$,  ${y}=\left( y_1,\hdots,y_N\right) ^\top$, $\varepsilon=(\varepsilon_1,\hdots,\varepsilon_N)^\top$, ${\eta }
_X\left( \theta \right) =\left( \eta \left( x_1,\theta \right)
,\hdots,\eta \left( x_N,\theta \right) \right) ^\top.$ The design space $\mathcal{X}$ is supposed here to be finite. Instead of the exact design $X$ we can
consider equivalently for any $x\in \mathcal{X}$ the value $\xi _X\left(
x\right) $ of the relative frequency of $x$ within the design $X$. By a
standard approximation procedure, we consider the set $\Xi $ of all
probability measures defined on $\mathcal{X}$, as the set of all approximate
designs allowed in the experiment.

In the main part of the present paper we suppose the
linearity of the response function $\eta \left( x_i,\theta \right)
=f^\top\left( x_i\right) \theta$, and we suppose $\Theta =\mathbb{R}^p$. In a standard way, to any $\xi \in \Xi $ is associated its information
matrix 
$$
M\left( \xi \right) =\sum_{x\in \mathcal{X}}f\left( x\right) f^\top\left( x\right) \xi
\left( x\right),
$$
with $f(x)=(f_1(x),\hdots,f_p(x))^\top$. According to the aim of the experiment, we may choose an optimality
criterion $\phi \left( \xi \right) $, and a design $\mu $ is considered $\phi $-optimal when $\phi \left( \mu \right) =\max_{\xi \in \Xi }\phi \left(
\xi \right) $. Standard criteria $\phi \left( \cdot\right) $ are concave
functions on $\Xi $ having a statistical interpretation.

In \cite{PP13} the criteria of $E$-, $c$-, and $G$-optimality were considered, and the
corresponding criteria functions have been rewritten in a form 
\[
\phi \left( \xi \right) =\min_{u\in \mathbb{R}^p}\sum_{x\in \mathcal{X}}T\left(
u,x\right) \xi \left( x\right) 
\]
with given $T\left( u,x\right)$. This, together with the standard
restrictions on $\xi ,$ defines an ``infinite-dimensional'' linear programming (LP) problem: to choose
the values of $\xi \left( x\right) ;x\in \mathcal{X}$ and of $t\in \mathbb{R}$ so to
maximize $t$ under infinitely many linear restrictions: 
\begin{eqnarray*}
\sum_{x\in \mathcal{X}}T\left( u,x\right) \xi \left( x\right) &\geq &t\text{
\quad for any }u\in \mathbb{R}^p ,\\
\xi \left( x\right) &\geq &0\quad \text{for any }x\in \mathcal{X}\text{, and}
\sum_{x\in \mathcal{X}}\xi \left( x\right) =1.
\end{eqnarray*}
In particular, for $E$-optimality, with $\phi _E\left( \xi \right) $ equal to
the minimum eigenvalue of $M\left( \xi \right)$, we have 
\[
\phi _E\left( \xi \right) =\min_{u\in \mathbb{R}^p}\frac{u^\top M\left( \xi \right)
u}{u^\top u }=\min_{u\in \mathbb{R}^p}\sum_{x\in \mathcal{X}}\frac{[f^\top\left(
x\right) u]^2}{u^\top u}\xi \left( x\right) .
\]
The main idea of \cite{PP14} was to substitute the nonlinear response function $
\eta \left( x,\theta \right) $ instead of $f^\top\left( x\right) \theta $ and
so to obtain new criteria for nonlinear models, with the aim
to detect the lack of identifiability under the design $\xi .$ However, a
second aim of \cite{PP14} was to point attention to the fact that for those expressions for criteria
an LP method could be used to obtain nearly optimum designs
in linear models.

In the present paper we follow this second aim, but for $D$-, $A$-, and $E_k$-optimality criteria and also for the computationally not easy task to find the
``criterion robust'' optimum design in linear models, or to find a $D$-optimum design under the condition that the $A$-optimality criterion exceeds a given value. 
The difficulties to achieve also the first aim for $D$-, $A$-, and $E_k$-criteria are discussed in Appendix.

We notice that LP method has been used to compute $c$-optimal design in \cite{HJ08} but under a quite different set-up.

\section{Reformulation of the optimality criteria}
\label{sec:sec1}
The $D$-optimal design maximizes $\det\left( M\left( \xi \right) \right) $,
hence minimizes the generalized variance of $\hat{\theta}$, the BLUE of $\theta$.
The $A$-optimal design minimizes the sum of the variances of $\hat{\theta}%
_1,\hdots,\hat{\theta}_p$. The $E_k$-optimal design maximizes the sum of the
smallest $k$ eigenvalues of $M\left( \xi \right) $. There are many forms of
expressing the corresponding criteria functions $\phi \left( \xi \right) $.
All forms of $\phi \left( \xi\right) $ representing the same
criterion maintain the ordering of the designs but differ by the scaling
of this ordering, say $\phi \left( \xi\right) =\ln \det \left[
M\left( \xi \right) \right] $ and $\phi \left( \xi \right)
=\det^{1/p}\left[ M\left( \xi \right) \right] $ for $D$-optimality, and
similarly for the other criteria. Here we prefer criteria functions which
are not only concave, but also positively homogeneous, $\phi \left( \alpha
\xi \right) =\alpha \phi \left( \xi \right) $ for $\alpha > 0$ (see \cite{P93} for a
justification). So for the $D$-optimality $\phi _D\left( \xi \right)
=\det^{1/p}\left[ M\left( \xi \right) \right] $, for the $A$-optimality $\phi
_A\left( \xi \right) =1/tr\left[ M^{-1}(\xi) \right] $ when $M\left(
\xi \right) $ is nonsingular, and for the $E_k$-optimality $\phi
_{E_k}\left( \xi \right) $ =$\lambda _1\left( \xi \right) +\hdots +\lambda
_k\left( \xi \right) $ where $\lambda _1\left( \xi \right) \leq \lambda
_2\left( \xi \right) \leq \hdots \leq \lambda _p\left( \xi \right) $ is the
ordering of eigenvalues of $M\left( \xi \right) $ respecting their
multiplicity. Denote $u_1\left( \xi \right) ,\hdots ,u_p\left( \xi\right) $ the
corresponding orthonormal eigenvectors of $M\left( \xi \right) $. Denote
also $\Xi ^{+}=\left\{ \mu \in \Xi :M\left( \mu \right) \text{ is nonsingular}\right\} $. $D$- and $A$-optimal designs are evidently localized on $\Xi ^{+}$,
what need not to be true for the $E_k$-optimality.
\begin{theorem}
\label{veta1}
We can write 
\begin{eqnarray}
\nonumber\phi _D\left( \xi \right) &=&\min_{\mu \in \Xi ^{+}}\sum_{x\in \mathcal{X}
}H_D(\mu,x) \xi \left(
x\right)  \\
&=&\min_{\mu \in \Xi ^{+}}\sum_{x\in \mathcal{X}
}\left\{ \frac{\det^{1/p}\left[ M\left( \mu \right) \right] }pf^\top \left(
x\right) M^{-1}\left( \mu \right) f\left( x\right) \right\} \xi \left(
x\right),\label{a}  \\
\nonumber \phi _A\left( \xi \text{ }\right) &=&
\min_{\mu \in \Xi ^{+}}\sum_{x\in \mathcal{X}
}H_A(\mu,x) \xi \left(
x\right) \\ 
& =&\min_{\mu \in \Xi ^{+}}\sum_{x\in 
\mathcal{X}}\left\{ \frac{\left\| M^{-1}\left( \mu \right) f\left( x\right)
\right\| ^2}{\left[tr\left( M^{-1}\left( \mu \right) \right)\right]^2 }\right\} \xi \left(
x\right)
 \label{b}
\end{eqnarray}
for any $\xi \in \Xi ^{+},$ and 
\begin{equation}
\phi _{E_k}\left( \xi \right) =\min_{\mu \in \Xi }\sum_{x\in \mathcal{X}
}H_{E_k}(\mu,x) \xi \left(
x\right)  =\min_{\mu \in \Xi }\sum_{x\in \mathcal{X}
}\left\| P^{(k)}\left( \mu \right) f\left( x\right) \right\| ^2\xi \left(
x\right)  \label{c}
\end{equation}
for any $\xi \in \Xi $. Here $P^{(k)}\left( \mu \right) $ is the $k$-dimensional
orthogonal projector $P^{(k)}\left( \mu \right) =\sum_{i=1}^ku_i\left( \mu
\right) u_i^\top\left(\mu \right) $, and $\left\| \cdot \right\|$ denotes the Euclidean norm.
\end{theorem}
\begin{proof}
In the proof we shall often use that $tr\left[ AB\right]
=tr\left[ BA\right] $ for any matrices $A=A_{l\times s},B=B_{s\times l}$ \citep{H00}. By
the known inequality between the geometric and arithmetic means of positive
numbers (cf. \cite[Chap.~2]{S04}), we obtain
\[
\left\{ \det \left[ S^\top M\left( \xi \right) S\right] \right\} ^{1/p}=\left\{
\Pi _{i=1}^p\alpha _i\right\} ^{1/p}\leq \frac 1p\sum_{i=1}^p\alpha _i=\frac
1p tr\left[ S^\top M\left( \xi \right) S\right] 
\]
for any nonsingular $p\times p$ matrix $S$. Here $\alpha _1,\hdots,\alpha _p$ are the
eigenvalues of $S^\top M\left( \xi \right) S$. So $\det^{1/p}\left[ M\left( \xi
\right) \right] \leq \frac 1p\det^{-1/p}[SS^\top]\sum_{x\in \mathcal{X}
}f^\top\left( x\right) SS^\top f\left( x\right) \xi \left( x\right) $, and we have
just to put $S=M^{-1/2}\left( \mu \right) $ to obtain the expression in~(\ref
{a}). If $S=M^{-1/2}\left( \xi \right) $, then $\alpha_i=1;$ $i=1\hdots p$, and
the geometric mean is equal to the arithmetic mean, so the minimum is
attained.

For any nonsingular $p\times p$ matrix $S$ we obtain from the Schwarz inequality 
\begin{eqnarray*}
\left[ tr\left( S\right) \right] ^2 &=&\left\{ tr\left[ M^{-1/2}\left( \xi
\right) SM^{1/2}\left( \xi \right) \right] \right\} ^2 \\
&\leq &tr\left[ M^{-1}\left( \xi \right) \right] tr\left[ M^{1/2}\left( \xi
\right) S^\top  S M^{1/2}\left( \xi \right) \right] =tr\left[ M^{-1}\left( \xi
\right) \right]tr\left[ SM\left( \xi \right) S^\top \right] 
\end{eqnarray*}
since in general $tr\left[ A^\top B\right] $ is a scalar product of matrices $A,B
$, and since $M^{-1/2}\left( \xi \right) $ and $M^{1/2}\left( \xi \right) $
are symmetric matrices. So $\left\{ tr\left[ M^{-1}\left( \xi \right)
\right] \right\} ^{-1}\leq tr\left[ SM\left( \xi \right) S^\top \right] /\left[
tr\left( S\right) \right] ^2=\sum_x\left\| Sf\left( x\right) \right\| ^2\xi
\left( x\right) /\left[ tr\left( S\right) \right] ^2$, and we have just to
put $S=M^{-1}\left( \mu \right) $ to obtain the expression in~(\ref{b}).
When $S=M^{-1}\left( \xi \right) $, we obtain evidently an equality in the
Schwarz inequality.

Denote $P=P^{(k)}\left( \mu \right)$. By the definition of  $
P^{(k)}\left( \mu \right) $ we have $PP=P$ and $P=P^\top$. So 
\[
\sum_{x\in \mathcal{X}}\left\| Pf\left( x\right) \right\| ^2\xi \left(
x\right) =tr\left[ PM\left( \xi \right) P\right]. 
\]

On the other hand, denote $U=\left( u_1\left( \xi \right) ,\hdots,u_p\left( \xi
\right) \right) ,\Lambda =diag\left\{ \lambda _1\left( \xi \right)
,\hdots,\lambda _p\left( \xi \right) \right\} $, and use that $M\left( \xi
\right) =U\Lambda U^\top$ to obtain 
\begin{eqnarray*}
tr\left[ PM\left( \xi \right) P\right] &=&tr\left[ PU\Lambda U^\top P\right]
=tr\left[ \Lambda \left( PU\right) ^\top \left( PU\right) \right] \\
&=&\sum_{i=1}^p\lambda _i\left( \xi \right) \left\{ \left( PU\right)
^\top \left( PU\right) \right\} _{ii}=\sum_{i=1}^p\lambda _i\left( \xi \right)
\left\| Pu_i\left( \xi \right) \right\| ^2=\sum_{i=1}^p\lambda _i\left( \xi
\right) w_i,
\end{eqnarray*}
where we denoted $w_i=\left\{ \left( PU\right) ^\top \left( PU\right) \right\}
_{ii}=\left\| Pu_i\left( \xi \right) \right\| ^2$. Since $UU^\top =U^\top U=I$, we
have 
\[
k=tr\left[ P\right] =tr\left[ P^\top P\right] =tr\left[ P^\top PUU^\top \right]
=\sum_{i=1}^p\left\{ \left( PU\right) ^\top \left( PU\right) \right\}
_{ii}=\sum_{i=1}^pw_i .
\]
Further $w_i\in [0,1]$, since $0\leq \left\| Pu_i\left( \xi
\right) \right\| ^2\leq \left\| u_i\left( \xi \right) \right\| ^2=1$. So,
using that $\lambda _1\left( \xi \right) \leq \hdots\leq \lambda _p\left( \xi
\right) $ we obtain that $\sum_{i=1}^p\lambda _i\left( \xi \right) w_i$ is
minimized exactly when the weights $w_i$ have maximum value $\left(
=1\right) $ at the smallest $k$ values of $\lambda _i\left( \xi \right) $.

 Summarizing we obtain 
\begin{equation}
\sum_{x\in \mathcal{X}}\left\| Pf\left( x\right) \right\| ^2\xi \left(
x\right) =tr\left[ PM\left( \xi \right) P\right] =\sum_{i=1}^p\lambda
_i\left( \xi \right) w_i\geq \sum_{i=1}^k\lambda _i\left( \xi \right) =\phi
_{E_k}\left( \xi \right) .
\label{c2}
\end{equation}
In the particular case that $P=P^{(k)}\left( \xi \right) =\sum_{j=1}^ku_j\left(
\xi \right) u_j^\top \left( \xi \right) $ we have $w_i=\left\| P^{(k)}\left( \xi
\right) u_i\left( \xi \right) \right\| ^2=\left\| u_i\left( \xi \right)
\right\| ^2=1$ if $i\leq k$, $\left\| P^{(k)}\left( \xi \right) u_i\left( \xi
\right) \right\| ^2=0$ if $i>k$, hence $\sum_{x\in \mathcal{X}}\left\|
P^{(k)}\left( \xi \right) f\left( x\right) \right\| ^2\xi \left( x\right)
=\sum_{i=1}^k\lambda _i\left( \xi \right) =\phi _{E_k}\left( \xi \right)$, which together with (\ref{c2}) yields an expression in (\ref{c}).

\end{proof}

\begin{remark}
\label{rem1}
 We could write in~(\ref{b}) $\phi _A\left( \xi \right) =\min_{B\in \mathcal{B}}\sum_{x\in \mathcal{X}}\left\{ \frac{\left\|
Bf\left( x\right) \right\| ^2}{\left[tr\left( B\right)\right]^2 }\right\} \xi \left(
x\right) $, where $\mathcal{B}$ is any set of nonsingular
matrices containing $M^{-1}\left( \xi \right) .$ When this formula should
hold for all $\xi \in \Xi $, then the set $\mathcal{B=}\left\{ M^{-1}\left(
\mu \right) :\mu \in \Xi ^{+}\right\} $ is the smallest of such sets. A
similar modification could be done for $D$-optimality in~(\ref{a}). In~(\ref{c}) we could minimize over any set of $k$-dimensional projectors containing $
P^{(k)}\left( \xi \right) $.
\end{remark}

\begin{remark}
\label{rem2}
As follows from \cite[Chap.~9.5]{PP13} we could obtain similar results
as in Theorem~\ref{veta1} by considering gradients or subgradients of $\phi \left( \xi
\right) $. However, the presented direct proofs, without using a not very
common notion of subgradients, can be more attractive for people in
applications.
\end{remark}

\section{The iterative computation by LP; the algorithms and examples}
\label{sec:sec2}
\subsection{Algorithm for $D$-, $A$-, and $E_k$-optimality}
\label{alg1}
Let us write $H(\mu,x )$ instead of $H_D(\mu,x)$, $H_A(\mu,x)$, or $H_{E_k}(\mu,x)$  from Theorem~\ref{veta1}. 
For the maximization of $\phi$ we apply a modification of the cutting-plane method \cite{K60} as presented in \cite{PP13} and \cite{PP14}:

\begin{enumerate}[start=0]
\item Take any vector $\xi^{(0)}$ such that $\sum_{x\in\mathcal{X}}\xi^{(0)}(x)=1$ and $\xi^{(0)}(x)\geq 0\;\forall\;x\in\mathcal{X}$, choose $\epsilon>0$, set $\Xi^{(0)}=\emptyset $ and $n=0$.
\item Set $\Xi^{(n+1)}= \Xi^{(n)} \cup \left\lbrace\xi^{(n)}\right\rbrace$.
\item Use the LP solver to find $\left(\xi^{(n+1)},t^{(n+1)}\right)$ so to maximize $t$ satisfying the constraints:
 \begin{itemize}
 \item $t>0,\; \xi(x)\geq0 \; \forall \; x\in\mathcal{X} ,\; \sum_{x\in\mathcal{X}}\xi(x)=1,$ 
 \item $\sum_{x\in\mathcal{X}} H(\mu,x)\xi(x)\geq t\; \forall \mu\in\Xi^{(n+1)}.$
 \end{itemize}
 \item Set $\Delta^{(n+1)}=t^{(n+1)}-\phi\left(\xi^{(n+1)}\right)$, if $\Delta^{(n+1)}<\epsilon$ take $\xi^{(n+1)}$ as an $\epsilon$-optimal design and stop, or else $n\leftarrow n+1$ and continue by step 1.
 \end{enumerate}

Notice that  $\min_{\mu\in\Xi^{(n+1)}}\sum_{x\in\mathcal{X}}H(\mu,x)\xi(x)$ is an upper piecewise linear approximation of $\phi(\xi)$. Increasing $n$, the set $\Xi^{(n+1)}\subseteq\Xi$ becomes larger and the approximation is better.

On the other hand, when $n$ is small, the information matrix $M\left(\xi^{(n)}\right)$ could be ill-conditioned or even singular. In order to avoid the difficulty with  inverse matrices  in $D$- and $A$-optimality, it is possible to use any symmetric positive definite matrix as a substitute for $M\left(\xi^{(n)}\right)$ as justified in Remark~\ref{rem1}. Alternatively, \cite{PP13} recommend the regularization $M\left(\xi^{(n)}\right)+\gamma I$, where $\gamma$ is a small positive number and $I$ is the identity matrix.
Note that it is also possible to take  $\Xi^{(0)}$ as an nonempty set containing $s\geq1$ initial designs. If $s$ or $n$ is large, the probability of ill-conditioned or singular information matrix $M\left(\xi^{(n)}\right)$ is less.

 The problem of singular information matrix does not appear in $E_k$-optimality criteria.

The stopping rule used in the above algorithm follows from the upper and lower bounds for $\max_{\xi\in\Xi}\phi(\xi)$:
$$
\phi\left(\xi^{(n+1)}\right)\leq \max_{\xi\in\Xi}\phi(\xi) \leq t^{(n+1)}.
$$
The first inequality is obvious.
 Note that $t^{(n+1)}=\max_{\xi\in \Xi}\min_{\mu\in\Xi^{(n+1)}}\sum_{x\in\mathcal{X}}H(\mu,x)\xi(x)$, while $\max_{\xi\in\Xi}\phi(\xi)=\max_{\xi\in\Xi}\min_{\mu\in\Xi}\sum_{x\in\mathcal{X}}H(\mu,x)\xi(x)$, and $\Xi\supseteq \Xi^{(n+1)}$.
This yields the second inequality.

There are also available stopping rules based on the equivalence theorem \citep{K74,KW59}, which are considered as standard. Let $\epsilon_{stop}$ be a chosen small nonnegative number.
An iterative algorithm will stop if $d\left(\xi^{(n)}\right)<\epsilon_{stop}$, where for $D$-optimality $d\left(\xi^{(n)}\right)=\left|\max_{x\in\mathcal{X}}f^\top(x)M^{-1}\left(\xi^{(n)}\right)f(x)-p\right|$ and  for the criterion of $A$-optimality
 $d\left(\xi^{(n)}\right)=\left|\max_{x\in\mathcal{X}}f^\top(x)M^{-2}\left(\xi^{(n)}\right)f(x)-tr\left[M^{-1}\left(\xi^{(n)}\right)\right]\right|$ as seen e.g. in \cite{K74,K75}. According to \cite{H04} the stopping rule for $E_k$-optimality criteria is 
 $d\left(\xi^{(n)}\right)=\left|\phi_{E_k}\left(\xi^{(n)}\right)- \max_{x\in\mathcal{X}} \sum_{i=1}^k \left[f^\top (x) u_i\left(\xi^{(n)}\right)\right]^2\right|$, which can be used only if  $\lambda_k\left(\xi^{(n)}\right)<\lambda_{k+1}\left(\xi^{(n)}\right)$.

As mentioned in \cite[Chap.~9.5]{PP13}, the cutting-plane method can have bad convergence properties (referenced to \cite{B06,N04}), one can then use the level method (see \cite{N04} or \cite{PP13}), which adds the quadratic programming step in the method of cutting planes.

In the examples below we compare the known optimal designs with results of our algorithm. The computations were performed in Matlab on a bi-processor PC (3.10 Ghz) equipped with 6GB of RAM and with 64 bits Windows 8.1. LP problems were solved with interior point method.

\begin{example}
\label{ex1}
Consider the nonlinear regression model of \cite{A93}.
\begin{equation*}
\eta(x,\theta)=\theta_1\left[exp(-\theta_2x)-exp(-\theta_3x)\right],\; x\in\mathbb{R}^+,\; \theta=(\theta_1,\theta_2,\theta_3)^\top.
\end{equation*}
 We use the algorithm of Sec.~\ref{alg1} to compute local $D$- and $E_1$-optimal designs for the nominal value of the parameter $\theta^0=(21.8,0.05884,4.298)^\top$, so we shall write $\partial\eta(x,\theta)/\partial\theta|_{\theta_0}$ instead of $f(x)$ everywhere.
We take a finite design space containing 24,000 points $\mathcal{X}=\{ 0.001, 0.002,\hdots,23.999,24.000\}$, $\epsilon=10^{-10}$ with $\xi^{(0)}(x)=1/3$ if $x\in \{0.2,1,23\}$ and $\xi^{(0)}(x)=0$ otherwise. The computed designs are given in Table~\ref{tabpr1}. Notice that the computed results correspond to those in \cite{A93}. 
\begin{table}[h!]
\centering
\begin{tabular}{cccccc}
  \hline
  $\phi$ & $\xi^*$ &$\phi^*$&iter. & time& $d(\xi^*)$\\
  \hline
   $D$ & \begin{tabular}{ccc}
  $0.229$&$1.389$&$18.417$\\
  $0.3333$&$0.3333$&$0.3333$\\
  \end{tabular}&$\phi_D^*=11.739$& 64 &16m 9s&$1.5 \cdot 10^{-5}$\\
    \hline \text{$E_1$} & \begin{tabular}{cccc}
  $0.169$&$1.394$&$23.402$&$23.403$\\
  $0.1993$&$0.6623$&$0.0415$&$0.0969$\\
  \end{tabular}&$\phi_{E_1}^*=0.3163$& 49 &5m 53s&$3.89\cdot 10^{-6}$\\
  \hline   
\end{tabular}
\caption{Example~\ref{ex1}: the locally optimal designs are $\xi^*_D$ and $\xi^*_{E_1}$ (column 2); $\phi^*_D=\phi_D(\xi^*_D)$ and $\phi_{E_1}^*=\phi_{E_1}(\xi^*_{E_1})$ (column 3); the number of iterations (column 4) and the computational time (column 5) required until the algorithm stopped; the  corresponding value of $d(\xi^*)$ based on the equivalence theorem (column 6).}
\label{tabpr1}
\end{table}
\end{example}

\subsection{Algorithm for computing criterion robust designs}
\label{alg2}

The criteria of $E_k$-optimality play
a special role in experimental design. We say that the design $\xi $ is not
worse than the design $\mu $ with respect to the Schur ordering of designs
if $\phi _{E_k}\left( \xi \right) \geq \phi _{E_k}\left( \mu \right) $ for
all $k=1,\hdots,p$. Then also $\phi \left( \xi \right) \geq \phi \left( \mu
\right) $ for many other optimality criteria. However, the Schur ordering is
a partial ordering of designs, and a Schur-optimal design exists only in some
very particular cases. On the other hand, if we denote by $\mathcal{O}$ the
set of all criteria functions $\phi \left( \xi \right) $, which are concave
and positive homogeneous, and moreover are orthogonally invariant in the
sense that $\phi \left( \xi \right) =\Phi \left[ M\left( \xi \right) \right] 
$ with $\Phi \left[ M\left( \xi \right) \right] $=$\Phi \left[ U^\top  M\left( \xi \right)  U\right] $ for any
orthogonal matrix $U,$ it makes sense to look for a design $\xi_{ef}$ which is
maximin efficient with respect to such criteria, i.e. 
\[
\xi_{ef} =\arg \max_{\xi \in \Xi }\min_{\phi \in \mathcal{O}}\left[ \frac{\phi
\left( \xi \right) }{\max_{\zeta \in \Xi }\phi \left( \zeta \right) }\right]. 
\]
Here the ratio $\frac{\phi \left( \xi \right) }{\max_{\zeta \in \Xi }\phi
\left( \zeta \right) }$ is called the $\phi$-efficiency of the design $\xi $. This maximin efficiency problem can be simplified (cf. \cite{H04}), the
solution $\xi_{ef}$ coincides with the solution of 
\[
\xi_{ef} =\arg \max_{\xi \in \Xi }\min_{1\leq k\leq p}\left[ \frac{\phi
_{E_k}\left( \xi \right) }{\max_{\zeta \in \Xi }\phi _{E_k}\left( \zeta
\right) }\right] ,
\]
i.e. with the design which is maximin efficient in the (finite) class of all
$E_k$-optimality criteria. Such a design is called also ``criterion robust''
in \cite{H04}. But even this problem is computationally difficult, mainly because the
$E_k$-optimality criteria are not differentiable. For us it is important that we can
approach the solution of this problem by the LP programming technique.
First, using Theorem~\ref{veta1} we compute $E_k\left( opt\right) =\max_{\zeta \in \Xi
}\phi _{E_k}\left( \zeta \right) $ for all $k$ (see Sec.~\ref{alg1}), and then we can formulate another ``infinite-dimensional'' LP problem: 
 to choose the values of $\xi(x);\;x\in\mathcal{X}$ and of $t\in\mathbb{R}$ so to maximize $t$ under linear constraints:
 \begin{eqnarray*}
\sum_{x\in \mathcal{X}}\frac{ H_{E_k}(\mu,x)}{E_k(opt)} \xi\left( x\right)  &\geq & t\text{  for any }\mu \in \Xi ^{+} \text{ and for every } k\in\{1,\hdots,p\},\\
\xi \left( x \right) &\geq & 0 \text{ for any } x\in\mathcal{X}, \text{ and } \sum_{x\in \mathcal{X}}\xi \left( x\right)  =1. 
\end{eqnarray*}
 
  In order to compute the maximin efficient design, the  algorithm of Sec.~\ref{alg1} needs to be modified in step 2. Actually, the constraints in the LP problem will be:

\begin{itemize}
 \item $t>0,\; \xi(x)\geq0 \; \forall \; x\in\mathcal{X},\;\sum_{x\in\mathcal{X}}\xi(x)=1, $ 
 \item $\sum_{x\in\mathcal{X}} \frac{H_{E_k}(\mu,x)}{E_k(opt)}\xi(x)\geq t\; \forall \mu\in\Xi^{(n+1)}$ and $\forall\;k=1,\hdots,p$,
\end{itemize}
where  $E_k(opt)=\max_{\zeta\in\Xi}\phi_{E_k}(\zeta)$ is computed using the unmodified algorithm of Sec.~\ref{alg1} for all $k=1,\hdots,p$.

\begin{example}
\label{ex2}
Consider the quadratic regression model on a $q$-dimensional cube:
\begin{equation}
\label{pr2}
y=\beta_0+\sum_{i=1}^q\beta_ix_i^2+\sum_{i=1}^q\beta^{(i)}x_i+\sum_{i<j}\beta_{ij}x_ix_j+\varepsilon,\; x=(x_1,\hdots,x_q)^\top\in [-1,1]^q
\end{equation}

with a parameter $\beta=(\beta_0,\beta_1,\hdots,\beta_q,\beta^{(1)},\hdots ,\beta^{(q)},\beta_{12},\hdots,\beta_{q-1,q})^\top$ of dimension $p=1+3/2q+q^2/2$.
The criterion robust design in the model~(\ref{pr2}) was analytically studied for $q=1$ in \cite{H04} and for $q=2$ in \cite{FH13}. The case of $q=3$ was numerically solved in \cite{FH13}.

Consider the set $C_i=\lbrace x\in\{-1,0,1\}^q:\;\sum_{j=1}^q|x_j|=i \rbrace$ for $i=0,1.\hdots q$. Thus, $C_0=\lbrace(0,\hdots,0)^\top\rbrace$ and $C_q$ is the set of all vertices of the $q$-dimensional cube. We shall denote $C=\bigcup_{i=0}^qC_i$ and $\xi(C_i)=\sum_{x\in C_i}\xi(x)$. As mentioned in \cite{FH13}, for every $\phi \in \mathcal{O}$ there exists a $\phi$-optimal design $\xi^*$ with support on $C$, such that for all $i=0,1,\hdots,q$ the measure $\xi^*(C_i)$  is uniformly distributed over points $x\in C_i$ (see also \cite{G87,H92}).

 Before computing the criterion robust designs, we needed to evaluate $E_k(opt)$ for $k=1,\hdots p$. The algorithm of Sec.~\ref{alg1} initialized with the uniform measure on $\mathcal{X}=C$ and with $\epsilon=10^{-10}$  gave the optimal values $E_k(opt)$ summarized in Table~\ref{tabpr2a} for $q=1,2,3, 4$. We observed the same optimal designs as calculated in \cite{H04} for $q=1$ and in \cite{FH13} for $q=2,3$.
  
 \begin{table}[h!]
\centering
\begin{tabular}{cccc}
  \hline
 $q$ & & &time\\
  \hline
  
   $1$ & \begin{tabular}{c} $k$ \\ $E_k(opt)$  \end{tabular} &\begin{tabular}{ccc}
  1&2&3\\ 
  $0.2$&$1$&$3$\\
    \end{tabular} &2s\\
    \hline  
   $2$ & \begin{tabular}{c} $k$ \\ $E_k(opt)$  \end{tabular} &\begin{tabular}{cccc}
  1&2&3 to 5&6\\ 
  $0.2$&$0.407$&$k-2$&$6$\\
  \end{tabular} &17s\\
    \hline
   $3$ & \begin{tabular}{c} $k$ \\ $E_k(opt)$  \end{tabular} &\begin{tabular}{cccccc}
  1&2&3&4&5 to 9&10\\ 
  $0.2$&$0.4$&$0.667$&$1.027$&$k-3$&$10$\\
    \end{tabular} &28m 17s\\
     \hline 
  $4$  & \begin{tabular}{c} $k$ \\ $E_k(opt)$  \end{tabular} &\begin{tabular}{ccccccc}
  1&2&3&4&5&6 to 14&15\\ 
  $0.2$&$0.0433$&$0.6242$&$0.4834$&$1.250$&$k-4$&15\\
   \end{tabular}  &23h 32m 47s\\
     \hline   
\end{tabular}
\caption{
Example~\ref{ex2}: the optimal values of the $E_k$-optimality criteria on a $q$-dimensional cube for the model~(\ref{pr2}) and the total computational time required until the optimal values  for all $k=1,\hdots ,p$ together were evaluated.}
\label{tabpr2a}
\end{table}
Then using the algorithm of Sec.~\ref{alg2} we computed criterion robust designs on $\mathcal{X}=C$ for $q=1,2,3,4$ obtaining the same results (except $q=4$) as in \cite{H04, FH13}, and the optimal mass concentrated on $C_i$ is listed in Table~\ref{tabpr2b}. Note, that for $q=3$ and $q=4$ the optimal design $\xi^*$ computed by algorithm of Sec.~\ref{alg2} does not put mass uniformly among $x\in C_i$ with $i=0,\hdots q$.
 By redistributing the mass $\xi^*(C_i)$ uniformly over $x\in C_i$ for $i=0,\hdots, q$, we obtained new design $\xi^{**}$ of identical  $\mathcal{O}$-minimal efficiency as achieved in $\xi^*$. Thus,  $\xi^{**}$ is another criterion robust design with required uniform measure on $C_i$ for any $i=0,\hdots q$. 
\begin{table}[h!]
\centering
\begin{tabular}{ccccc}
  \hline
  $q$& $\xi^*$ &$\Psi^*$&iter. & time\\
  \hline
 
  $1$ & \begin{tabular}{cc}
  $\xi^*(C_0)$&$\xi^*(C_1)$\\
  $0.3532$&$0.6468$\\
  \end{tabular}&$0.7646$& 16&1s\\
   \hline
  $2$ & \begin{tabular}{ccc}
  $\xi^*(C_0)$&$\xi^*(C_1)$&$\xi^*(C_2)$\\
  $0.1775$&$0.2924$&$0.5304$\\
  \end{tabular}&$0.7060$& 108&1m 20s\\
  \hline
  $3$ & \begin{tabular}{cccc}
  $\xi^*(C_0)$&$\xi^*(C_1)$&$\xi^*(C_2)$&$\xi^*(C_3)$\\
  $0.0884$&$0.2343$&$0.2306$&$0.4467$\\
  \end{tabular}&$0.6642$&464&17m 9s\\
  \hline
  $4$ & \begin{tabular}{ccccc}
  $\xi^*(C_0)$&$\xi^*(C_1)$&$\xi^*(C_2)$&$\xi^*(C_3)$&$\xi^*(C_4)$\\
  $0.1097$&$0.0559$&$0.1437$&$0.3062$&$0.3845$\\
  \end{tabular}&$0.6526$&1453&10h 20m 7s\\
  \hline   
\end{tabular}
\caption{Example~\ref{ex2}: criterion robust designs $\xi^*$ (column 2) and the $\mathcal{O}$-minimal efficiency of $\xi^*$, i.e. $\Psi^*=\max_{\xi\in\Xi}\min_k\frac{\phi_{E_k}(\xi^*)}{E_k(opt)}$ (column 3) on a $q$-dimensional cube for the model~(\ref{pr2});  the number of iterations (column 4)  and the computational time (column 5) required until the algorithm stopped.}
\label{tabpr2b}
\end{table}

Alternatively, we  computed the criterion robust design for $q=2$ (thus $p=6$)  on a modified design space 
$\mathcal{X}=\{-1,-0.95,\hdots ,0.9,0.95,1\}^2$ (i.e. $\mathcal{X}$ is a grid consisting of 1,681 two-dimensional points including the set $C$). Assuming that the values $E_k(opt)$ are known  or previously computed for all $k \in{1,\hdots,p}$, the algorithm of Sec.~\ref{alg2} initialized with uniform measure on $\mathcal{X}$ and  $\epsilon=10^{-10}$ converged after 102 iterations in 36m  and 5s with the same results as given in Table~\ref{tabpr2b}.
\end{example}

\subsection{Algorithm for $D$-optimality conditioned by prescribed $A$-optimality}
\label{alg3}

It is not difficult to see that in the considered LP
problems we can easily add some supplementary constraints linear in $\xi $,
say a cost constraint $\sum_{x\in \mathcal{X}}c\left( x\right) \xi \left(
x\right) =c$, where $c\left( x\right) $ is the cost of an observation at $x$
and $c$ is proportional to the total cost allowed for the whole experiment.
What is less evident is that we can combine optimality criteria. Say, when
we want to obtain a $D$-optimal design under the condition that the
$A$-optimality criterion attains a prescribed value $a$, we have to solve the
``infinite-dimensional'' LP problem: to choose the values of $\xi(x);\;x\in\mathcal{X}$ and of $t\in\mathbb{R}$ so to maximize $t$ under linear constraints:
\begin{eqnarray*}
\sum_{x\in \mathcal{X}} H_D(\mu,x)\xi \left( x\right)  &\geq & t\text{  for any }\mu \in \Xi ^{+},\\
\sum_{x\in \mathcal{X}} H_A(\mu,x)\xi \left( x\right)  &\geq & a\text{  for any }\mu \in \Xi ^{+}, \\
\xi \left( x \right) &\geq & 0 \text{ for any } x\in\mathcal{X}, \text{ and } \sum_{x\in \mathcal{X}}\xi \left( x\right)  =1. 
\end{eqnarray*}
This problem can be solved by the algorithm of Sec.~\ref{alg1} with a modification in constraints of the LP problem and in the stopping rule. 

\begin{enumerate}[start=0]
\item Take any vector $\xi^{(0)}$ such that $\sum_{x\in\mathcal{X}}\xi^{(0)}(x)=1$ and $\xi^{(0)}(x)\geq 0\;\forall\;x\in\mathcal{X}$, choose $\epsilon_D>0$, $\delta_A \approx 0$ , set $\Xi^{(0)}=\emptyset $ and $n=0$.
\item Set $\Xi^{(n+1)}= \Xi^{(n)} \cup \left\lbrace\xi^{(n)}\right\rbrace$.
\item Use the LP solver to find $\left(\xi^{(n+1)},t^{(n+1)}\right)$ so to maximize $t$ satisfying the constraints:
\begin{itemize}
 \item $t>0,\; \xi(x)\geq0 \; \forall \; x\in\mathcal{X}, \; \sum_{x\in\mathcal{X}}\xi(x)=1 ,$ 
 \item $\sum_{x\in\mathcal{X}} H_D(\mu,x)\xi(x)\geq t\; \forall \mu\in\Xi^{(n+1)},$
 \item $\sum_{x\in\mathcal{X}} H_A(\mu,x)\xi(x)\geq a\; \forall \mu\in\Xi^{(n+1)}.$
 \end{itemize}
 \item Set $\Delta_D^{(n+1)}=t^{(n+1)}-\phi_D\left(\xi^{(n+1)}\right)$ and $\Delta_A^{(n+1)}=\phi_A\left(\xi^{(n+1)}\right)-a$. If $\Delta_D^{(n+1)}<\epsilon_D $ and $ \Delta_A^{(n+1)}>\delta_A $ take $\xi^{(n+1)}$ as an $(\epsilon_D,\delta_A)$-optimal design and stop, or else $n\leftarrow n+1$ and continue by step 1.

 \end{enumerate}

The constant $\delta_A$ is chosen at the beginning of the algorithm. The preferred value is $\delta_A=0$, however choosing $\delta_A<0$ but small, we can reduce the strictness of the condition on $A$-optimality. 

Now consider the set $\mathcal{A}^{(n)}=\{\xi\in\Xi: \sum_{x\in\mathcal{X}}H_A(\mu,x)\xi(x)\geq a\; \forall \mu \in \Xi^{(n)}\}$, then $\mathcal{A}^{(n)} \supset \mathcal{A}^{(n+1)} \supset \mathcal{A}=\{\xi\in\Xi: \phi_A(\xi)\geq a\}$. So the exact  solution of our problem would be $\xi^*=\arg\max_{\xi \in \mathcal{A}}\phi_D(\xi)$. We  can write:
$$
\begin{aligned}
t^{(n+1)}&=\max_{\xi \in \mathcal{A}^{(n+1)}}\min_{\mu\in\Xi^{(n+1)}}\sum_{x\in\mathcal{X}}H_D(\mu,x)\xi(x)\\
&\geq \max_{\xi\in\mathcal{A}^{(n+1)}}\min_{\mu \in \Xi}\sum_{x\in\mathcal{X}}H_D(\mu,x)\xi(x)=\max_{\xi\in\mathcal{A}^{(n+1)}} \phi_D(\xi),
\end{aligned}
$$
and then
\begin{equation}
\label{d1}
\max_{\xi \in \mathcal{A}}\phi_D(\xi)\leq \max_{\xi \in \mathcal{A}^{(n+1)}}\phi_D(\xi)\leq t^{(n+1)},
\end{equation}
\begin{equation}
\label{d2}
\phi_D\left(\xi^{(n+1)}\right)\leq \max_{\xi \in \mathcal{A}^{(n+1)}}\phi_D(\xi)\leq t^{(n+1)},
\end{equation}
where
$$
\xi^{(n+1)}=\arg \max_{\xi \in \mathcal{A}^{(n+1)}}\min_{\mu\in\Xi^{(n+1)}}\sum_{x\in\mathcal{X}}H_D(\mu,x)\xi(x).
$$
Assume that $\delta_A=0$ and the algorithm stopped, i.e. $t^{(n+1)}-\phi_D\left(\xi^{(n+1)}\right)<\epsilon_D$ and $\phi_A\left(\xi^{(n+1)}\right)\geq a$. 
According to~(\ref{d1}) and~(\ref{d2}) there are only two possibilities: first, if $\max_{\xi \in \mathcal{A}}\phi_D(\xi)\leq \phi_D\left(\xi^{(n+1)}\right)\leq t^{(n+1)}$, then  $\xi^{(n+1)}\in\mathcal{A}^{(n+1)}$ is even ``better'' design than we expected; second, $ \phi_D\left(\xi^{(n+1)}\right)\leq\max_{\xi \in \mathcal{A}}\phi_D(\xi)\leq t^{(n+1)}$, and the stopping rule implies that $\max_{\xi \in \mathcal{A}}\phi_D(\xi)- \phi_D\left(\xi^{(n+1)}\right)<\epsilon_D$, thus $\xi^{(n+1)}$ is an $\epsilon_D$-optimal design in both cases.

\begin{example}
\label{ex3}
Consider the polynomial regression model of degree $d$:
$$
{y}=\theta_0+\theta_1 x+\theta_2 x^2 + \hdots +\theta_d x^d + \varepsilon,\; x\in [-1,1],\; \theta=(\theta_0,\theta_1\hdots,\theta_d)^\top.
$$
Denote by $\xi^*_{D|a}$ the $D$-optimal design under the condition that the $A$-criterion exceeds a value $a$. 
Set $\mathcal{X}=\{-1.00,-0.99,-0.98,\hdots,0.99,1.00\}$  as the  design space, suppose that the initial design $\xi^{(0)}$ allocates the unit mass uniformly to each $x \in \mathcal{X}$,  $\epsilon_D=10^{-10}$, and $\delta_A=0$. 
In Table~\ref{tabpr3} are given optimal designs $\xi^*_{D|a}$ for some particular values of $a$ and for $d=4$ computed by the algorithm of Sec.~\ref{alg3} with abovementioned setting. Notice that the $D$- and $A$-optimal (maximum) values are $\phi_D^*=0.1339$ and $\phi_A^*=0.0053$ respectively (see the optimal designs in polynomial regression in \cite[Chap.~11]{AD92} and \cite{PT91}). When $a$ is small, the algorithm of Sec.~\ref{alg3} will compute the $D$-optimal design. The initial knowledge of $\phi_A^*$ is necessary  because if  $a$  exceeds $\phi_A^*$, the algorithm does not work.
 Figure~\ref{Obr} displays $\phi_D$ and $\phi_A$ efficiencies of $\xi^*_{D|a}$ as a function of $a$, i.e. $\text{eff}_D(a)=\phi_D(\xi^*_{D|a})/\phi^*_D$ and $\text{eff}_A(a)=\phi_A(\xi^*_{D|a})/\phi^*_A$.

\begin{table}[h!]
\centering
\begin{tabular}{cccccc}
  \hline
  $a$& $\xi^*_{D|a}$ &$\phi^*_{D|a}$&$\phi_A(\xi^*_{D|a})$&iter. & time\\
  \hline
   0.0052 & \begin{tabular}{ccccc}
  $-1$&$-0.68$&$0$&$0.68$&$1$\\
  $0.136$&$0.2338$&$0.2604$&$0.2338$&$0.136$\\
  \end{tabular}&$0.1283$&$0.0052$& 97 & 67s\\
  \hline
   0.005 & \begin{tabular}{ccccc}
  $-1$&$-0.68$&$0$&$0.68$&$1$\\
  $0.1623$&$0.2194$&$0.2366$&$0.2194$&$0.1623$\\
  \end{tabular}&$0.1317$&$0.005$& 97 & 53s\\
  \hline
 0.002 & \begin{tabular}{ccccccc}
  $-1$&$-0.66$&$-0.65$&$0$&$0.65$&$0.66$&$1$\\
  $0.2$&$0.0847$&$0.1152$&$0.2$&$0.1151$&$0.0850$&$0.2$\\
  \end{tabular}&$0.1338$&$0.0044$& 163 & 158s\\
  \hline
\end{tabular}
\caption{Example~\ref{ex3}: the optimal designs $\xi^*_{D|a}$ (column 2) with different choices of $a$ (column 1); $\phi^*_{D|a}=\phi_D(\xi^*_{D|a})$- the value of the $D$-optimality criterion (column 3); $\phi_A(\xi^*_{D|a})$- the value of the $A$-optimality criterion (column 4); the number of iterations (column 5) and the computational time (column 6) required until the algorithm stopped.}
\label{tabpr3}
\end{table}

\begin{figure}[h!]
    \centering
    \includegraphics[width=0.4\textwidth, angle =270 ]{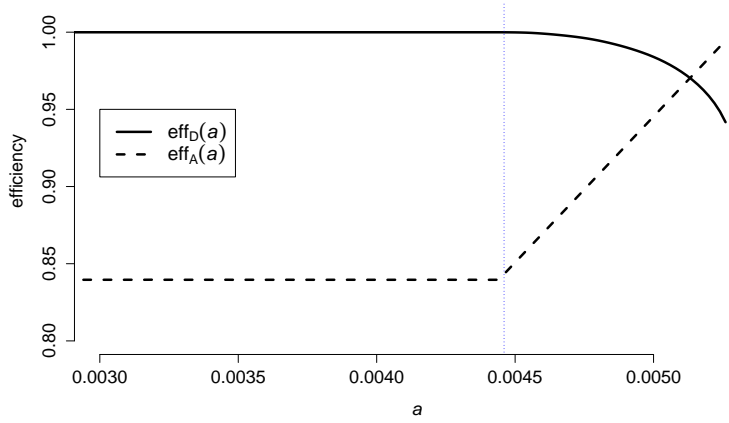}
    \caption{The graph of $\phi_A$-efficiency (dashed line) and of $\phi_D$-efficiency (solid line) of $\xi^*_{D|a}$ as a function of prescribed value $a$ in Example~\ref{ex3}.}
    \label{Obr}
\end{figure}
\end{example}

\section{Reformulation of AVE criteria in nonlinear experiments}
\label{sec:sec3}
In general, the information matrix in nonlinear regression model ${y}=\eta_X(\theta)+\varepsilon$ is a function of the parameter $\theta$. Similarly as in  Theorem~\ref{veta1},
we rewrite (local) $D$-, $A$-, and $E_k$-optimality criteria  in nonlinear regression model to a form:
\begin{equation}
\label{eqv1}
\phi(\xi,\theta)=\min_{\mu\in\Xi^*}\sum_{x\in\mathcal{X}}H(\mu,x,\theta)\xi(x).
\end{equation}
Here  $\Xi^*$ can be replaced by $\Xi$ or $\Xi^+$ depending on the considered criterion like in Theorem~\ref{veta1}.
The reformulation of expressions in Theorem~\ref{veta1} in terms of average (AVE) optimality  criteria $\int_{\Theta}\phi(x,\theta)d\pi(\theta)$, where $\pi(\theta)$ is supposed to be known prior distribution, is also possible and is given in Theorem~\ref{veta3}.
\begin{theorem}
We can write
$$
\int_{\Theta}\phi(\xi,\theta)d\pi(\theta)=\min_{\mu\in\Xi^*}\sum_{x\in\mathcal{X}}K(\mu,x,\theta)\xi(x),   
$$
where $K(\mu,x,\theta)=\int_{\Theta}H(\mu,x,\theta)d\pi(\theta)$.

\label{veta3}
\end{theorem}
\begin{proof}
The design space $\mathcal{X}$ is assumed to be finite, hence the summation and the integration are interchangeable. From~(\ref{eqv1}) we have  $\phi(\xi,\theta)\leq\sum_{x\in\mathcal{X}}H(\mu,x,\theta)\xi(x)$ for any $\mu\in\Xi^*$ and for all $\theta\in\Theta$. We can write
\begin{equation}
\label{eqv2}
\int_{\Theta}\phi(\xi,\theta)d\pi(\theta)\leq\sum_{x\in\mathcal{X}}\left[\int_{\Theta}H(\mu,x,\theta)d\pi(\theta)\right]\xi(x).
\end{equation}
  Since  the inequality~(\ref{eqv2}) holds for every $\mu\in\Xi^*$, evidently:
\begin{equation}
\label{eqv3}
 \int_{\Theta}\phi(\xi,\theta)d\pi(\theta)\leq \min_{\mu\in\Xi^*}\sum_{x\in\mathcal{X}}\left[\int_{\Theta}H(\mu,x,\theta)d\pi(\theta)\right]\xi(x).
\end{equation}
   Theorem~\ref{veta1} implies that minimum is in~(\ref{eqv1}) attained at $\mu=\xi$  for any $\theta\in\Theta$,  so we obtain an equality in~(\ref{eqv2}) for $\mu=\xi$, which together with~(\ref{eqv3}) proofs the theorem.
\end{proof}

\section*{Appendix: Reformulation of criteria in terms of nonlinear models}
\label{sec:sec4}

Using the notation $\eta \left( x,\theta \right) =f^\top\left( x\right) \theta $
we can rewrite the expressions from Theorem~\ref{veta1} to a form, which formally
allows an extension of criteria to a nonlinear model 
\begin{eqnarray*}
y_x &=&\eta \left( x,\theta \right) +\varepsilon _x ,\\
\theta &\in &\Theta \subset \mathbb{R}^p.
\end{eqnarray*}
However, for the $D$-, $A$-, and $E_k$-optimality criteria we are not so successful as for the $E$-, $c$-, and $G$-optimality criteria in \cite{PP14}. Therefore we put the corresponding constructions only in the Appendix. 
\begin{theorem}
\label{veta2}
Let $\theta ^{\left( 0\right) }\in \Theta $ be a given vector. Denote 
$$
\mathcal{V}_{\theta ^{\left( 0\right) }}=\left\{ \left( \theta ^{\left(
1\right) },\hdots,\theta ^{\left( p\right) }\right) :\forall _i\;\theta
^{\left( i\right) }\in \Theta ,\,\left( \theta ^{\left( i\right) }-\theta
^{\left( 0\right) }\right) \neq 0\text{, }\left( \theta ^{\left( i\right)
}-\theta ^{\left( 0\right) }\right) ^\top\left( \theta ^{\left( j\right)
}-\theta ^{\left( 0\right) }\right) =0\text{, }i\neq j\right\} 
.$$
Further denote by $\left\| \theta ^{\left( i\right) }-\theta ^{\left(
0\right) }\right\| $ the Euclidean norm of $\theta ^{\left( i\right)
}-\theta ^{\left( 0\right) },$ and 
\[
\left\| \eta \left( \cdot,\theta ^{\left( i\right) }\right) -\eta \left(
\cdot,\theta ^{\left( 0\right) }\right) \right\| _\xi ^2=\sum_{x\in \mathcal{X}%
}\left[ \eta \left( x,\theta ^{\left( i\right) }\right) -\eta \left(
x,\theta ^{\left( 0\right) }\right) \right] ^2\xi \left( x\right) 
.\]
The ``extended'' criteria defined as:  
$$
\begin{aligned}
\phi _{eD}\left( \xi \right)& =\min_{\left( \theta ^{\left( 1\right)
},\hdots,\theta ^{\left( p\right) }\right) \in \mathcal{V}_{\theta ^{\left(
0\right) }}}\frac{\left( 1/p\right) \sum_{i=1}^p\left\| \eta \left( \cdot,\theta
^{\left( i\right) }\right) -\eta \left( \cdot,\theta ^{\left( 0\right) }\right)
\right\| _\xi ^2}{\left[ \prod_{j=1}^p\left\| \theta ^{\left( j\right)
}-\theta ^{\left( 0\right) }\right\| ^2\right] ^{1/p}},\\
\phi _{eA}\left( \xi \right) &=\min_{\left( \theta ^{\left( 1\right)
},\hdots,\theta ^{\left( p\right) }\right) \in \mathcal{V}_{\theta ^{\left(
0\right) }}}\frac{\sum_{i=1}^p\left\| \theta ^{\left( i\right) }-\theta
^{\left( 0\right) }\right\| ^2\left\| \eta \left( \cdot,\theta ^{\left( i\right)
}\right) -\eta \left( \cdot,\theta ^{\left( 0\right) }\right) \right\| _\xi ^2}{%
\left[ \sum_{j=1}^p\left\| \theta ^{\left( j\right) }-\theta ^{\left(
0\right) }\right\| ^2\right] ^2},\\
\phi _{eE_k}\left( \xi \right) &=\min_{\left( \theta ^{\left( 1\right)
},\hdots,\theta ^{\left( p\right) }\right) \in \mathcal{V}_{\theta ^{\left(
0\right) }}}\sum_{i=1}^k\frac{\left\| \eta \left( \cdot,\theta ^{\left( i\right)
}\right) -\eta \left( \cdot,\theta ^{\left( 0\right) }\right) \right\| _\xi ^2}{%
\left\| \theta ^{\left( i\right) }-\theta ^{\left( 0\right) }\right\| ^2}
\end{aligned}
$$
coincide with those in Theorem~\ref{veta1} in case that the model is linear.
\end{theorem}

\begin{proof}

Consider first the expression for $\phi _D\left( \xi \right) $ in Theorem~\ref{veta1}.
Using the notation from Sec.~\ref{sec:sec1} for every $\mu
\in \Xi ^{+}$ we can write $M^{-1}\left( \mu \right) =\sum_{i=1}^p\nu
_i\left( \mu \right) \nu _i^\top\left( \mu \right) $ with $\nu _i\left( \mu
\right) $ $=u_i\left( \mu \right) /\sqrt{\lambda _i\left( \mu \right) }$
(the normed eigenvector divided by the square root of the eigenvalue), and $%
\left\| \nu _i\left( \mu \right) \right\| ^2=\lambda _i^{-1}\left( \mu
\right) $. It follows that 
\[
\frac{\det^{1/p}\left[ M\left( \mu \right) \right] }pf^\top\left( x\right)
M^{-1}\left( \mu \right) f\left( x\right) =\frac{\left( 1/p\right)
\sum_{i=1}^p\left[ f^\top\left( x\right) \nu _i\left( \mu \right) \right] ^2}{%
\left[ \prod_{j=1}^p\left\| \nu _i\left( \mu \right) \right\| ^2\right]
^{1/p}} 
.\]
Denote $\theta ^{\left( i\right) }\left( \mu \right) =\theta ^{\left( 0\right)
}+\nu _i\left( \mu \right) $. In the linear model $f^\top\left(
x\right) \nu _i\left( \mu \right) =\eta \left( x,\theta ^{\left( i\right)
}\left( \mu \right) \right) -\eta \left( x,\theta ^{\left( 0\right) }\left(
\mu \right) \right) $.  So from Theorem~\ref{veta1} it follows that 
\begin{equation}
\phi _D\left( \xi \right) =\min_{\mu \in \Xi ^{+}}\frac{\left( 1/p\right)
\sum_{i=1}^p\left\| \eta \left(\cdot,\theta ^{\left( i\right) }\left( \mu
\right) \right) -\eta \left( \cdot,\theta ^{\left( 0\right) }\left( \mu \right)
\right) \right\| _\xi ^2}{\left[ \prod_{j=1}^p\left\| \theta ^{\left(
j\right) }\left( \mu \right) -\theta ^{\left( 0\right) }\left( \mu \right)
\right\| ^2\right] ^{1/p}}  .\label{e}
\end{equation}

Evidently $\left( \theta ^{\left( 1\right) }\left( \mu \right) ,\hdots,\theta
^{\left( p\right) }\left( \mu \right) \right) \in \mathcal{V}_{\theta
^{\left( 0\right) }}$. On the other hand, for any $\left( \theta ^{\left(
1\right) },\hdots,\theta ^{\left( p\right) }\right) \in \mathcal{V}_{\theta
^{\left( 0\right) }}$ we define  $B=\left[\sum_{i=1}^p\left(
\theta ^{\left( i\right) }-\theta ^{\left( 0\right) }\right) \left( \theta
^{\left( i\right) }-\theta ^{\left( 0\right) }\right) ^\top\right]^{-1}$. From Remark~\ref{rem1} of Theorem~\ref{veta1} it follows that we can take the minimum in~(\ref{e}) with
respect to all $\left( \theta ^{\left( 1\right) },\hdots,\theta ^{\left(
p\right) }\right) \in \mathcal{V}_{\theta ^{\left( 0\right) }}$ and not with respect to all $\mu\in\Xi^+$.

We proceed similarly for $A$-optimality. We have $M^{-2}\left( \mu \right)
=\sum_{i=1}^P\left\| \nu _i\left( \mu \right) \right\| ^2\nu _i\left( \mu
\right) \nu _i^\top\left( \mu \right) $ and $tr\left[ M^{-1}\left( \mu
\right) \right] =\sum_{i=1}^p\lambda _i^{-1}\left( \mu \right) $, so 
\begin{eqnarray*}
\frac{\left\| M^{-1}\left( \mu \right) f\left( x\right) \right\| ^2}{\left\{
tr\left[ M^{-1}\left( \mu \right) \right] \right\} ^2} &=&\frac{%
\sum_{i=1}^p\left\| \nu _i\left( \mu \right) \right\| ^2\left[ f^\top\left(
x\right) \nu _i\left( \mu \right) \right] ^2}{\left[ \sum_{j=1}^p\left\| \nu
_j\left( \mu \right) \right\| ^2\right] ^2} \\
&=&\frac{\sum_{i=1}^p\left\| \theta ^{\left( i\right) }\left( \mu \right)
-\theta ^{\left( 0\right) }\right\| ^2\left[ \eta \left( x,\theta ^{\left(
i\right) }\left( \mu \right) \right) -\eta \left( x,\theta ^{\left( 0\right)
}\left( \mu \right) \right) \right] ^2}{\left[ \sum_{j=1}^p\left\| \theta
^{\left( j\right) }\left( \mu \right) -\theta ^{\left( 0\right) }\right\|
^2\right] ^2}.
\end{eqnarray*}

For the $E_k$-optimality criterion we write $P^{(k)}\left( \mu \right)
=\sum_{i=1}^k\left\| \nu _i\left( \mu \right) \right\| ^{-2}\nu _i\left( \mu
\right) \nu _i^\top\left( \mu \right) $, hence 
\begin{eqnarray*}
\left\| P^{(k)}\left( \mu \right) f\left( x\right) \right\| ^2
&=&\sum_{i=1}^k\left\| \nu _i\left( \mu \right) \right\| ^{-2}\left[
f^\top\left( x\right) \nu _i\left( \mu \right) \right] ^2 \\
&=&\sum_{i=1}^k\frac{\left[ \eta \left( x,\theta ^{\left( i\right) }\left(
\mu \right) \right) -\eta \left( x,\theta ^{\left( 0\right) }\left( \mu
\right) \right) \right] ^2}{\left\| \theta ^{\left( i\right) }\left( \mu
\right) -\theta ^{\left( 0\right) }\right\| ^2}.
\end{eqnarray*}

\end{proof}

\begin{remark} 
\label{rem3}
The expressions in Theorem~\ref{veta2} are evidently linear in $\xi $, so maximization of $\phi _D\left( \xi \right) ,\,\phi _A\left(
\xi \right) ,$ and $\phi _{E_k}\left( \xi \right) $ with respect to $\xi $
corresponds to an ``infinite-dimensional'' LP problem even in a nonlinear
model. However this problem is too complex to be used for experimental
design. Moreover, in contrast to the criteria considered in \cite{PP14}, a clear
statistical interpretation is still missing. 
\end{remark}

\paragraph{Acknowledgements.} We would like to thank Luc Pronzato for  helpful advises. The paper was supported by the Slovak VEGA-Grant No. 1/0163/13.

\bibliographystyle{apalike}
\bibliography{references}

\end{document}